\newcommand*{\myTitle}{Monitoring of Domain-Related Problems in Distributed Data Streams}

\documentclass[runningheads,a4paper]{article}

\usepackage{amssymb}
\usepackage{amsmath}
\usepackage{amsthm}
\usepackage{hyperref}
\usepackage{graphicx}
\usepackage{mathtools}

\usepackage{nicefrac}
\usepackage{xspace}

\usepackage{enumitem}

\usepackage{algpseudocode}
\usepackage{algorithm}

\usepackage{todonotes}
\usepackage{url}
\urldef{\mailsa}\path|{pbemmann, felixbm, jbuerman, kempera, tillk, stknorr,
	nkothe, amaecker, malatya, fmadh, soerenri, jschaef, janniksu}@mail.uni-paderborn.de|

\usepackage{enumitem}

\usepackage[capitalise,nameinlink,noabbrev]{cleveref}
\crefname{step}{Step}{Steps}
\creflabelformat{step}{#2#1#3}

\newcommand{\set}[1]{\left\{ #1 \right\}}

\newcommand{\paren}[1]{\left( #1 \right)}

\newcommand{\bigO}{\mathcal{O}}

\newcommand{\Prob}[1]{\ensuremath{\Pr \left[ #1 \right]}}

\newcommand{\E}{\mathbb{E}}

\newtheorem{theorem}{Theorem}[section]
\newtheorem{lemma}[theorem]{Lemma}
\newtheorem{corollary}[theorem]{Corollary}

\title{\myTitle\thanks{This work was partially supported by the German Research Foundation (DFG) within the Priority Program ``Algorithms for Big Data'' (SPP 1736) and by the Federal Ministry of Education and Research (BMBF) as part of the poject ``Resilience by Spontaneous Volunteers Networks for Coping with Emergencies and Disaster'' (RESIBES), (grant no 13N13955 to 13N13957).}}

\author{Pascal Bemmann \and Felix Biermeier \and Jan B\"urmann \and Arne Kemper \and Till Knollmann \and Steffen Knorr \and Nils Kothe \and Alexander M\"acker \and Manuel Malatyali \and Friedhelm Meyer auf der Heide \and S\"oren Riechers \and Johannes Schaefer \and Jannik Sundermeier
	\\ [0.4em]
	Heinz Nixdorf Institute \& 	Computer Science Department\\
	Paderborn University, Germany\\[0.2em]
	\{pbemmann, felixbm, jbuerman, kempera, tillk, stknorr,\\
	nkothe, amaecker, malatya, fmadh, soerenri, jschaef, janniksu\}\\
	@mail.uni-paderborn.de
}
\date{} 

\newcommand{\comment}[2]{}

\renewcommand{\todo}[1]{\comment{todo}{#1}}
\newcommand{\ntv}{N_t^v}
\newcommand{\ntvsize}{|\ntv|}
\newcommand{\edapproximation}{($\varepsilon$,$\delta$)-approximation\xspace}

\newcommand{\protone}{\textsc{ConstantResponse}}

\newcommand{\CoFac}{\textsc{ConstantFactorApproximation}}
\newcommand{\CoFacS}{\CoFac\xspace}

\newcommand{\ntildev}{\tilde n^v}
\newcommand{\nbarv}{\bar n^v}

\newcommand{\EpsFac}{\textsc{EpsilonFactorApprox}}
\newcommand{\EpsFacS}{\EpsFac\xspace}

\newcommand{\ContEpsFac}{\textsc{ContinuousEpsilonApprox}}
\newcommand{\ContEpsFacS}{\ContEpsFac\xspace}

\newcommand{\DoMon}{\textsc{DomainMonitoring}}
\newcommand{\DoMonS}{\DoMon\xspace}

\begin{document}
	
\maketitle
	
\begin{abstract}
Consider a network in which $n$ distributed nodes are connected to a single server.
Each node continuously observes a data stream consisting of one value per discrete time step. 
The server has to continuously monitor a given parameter defined over all information available at the distributed nodes.
That is, in any time step $t$, it has to compute an output based on all values currently observed across all streams.
To do so, nodes can send messages to the server and the server can broadcast messages to the nodes.
The objective is the minimisation of communication while allowing the server to compute the desired output.

We consider monitoring problems related to the domain $D_t$ defined to be the set of values observed by at least one node at time $t$.
We provide randomised algorithms for monitoring $D_t$, (approximations of) the size $|D_t|$ and the frequencies of all members of $D_t$.
Besides worst-case bounds, we also obtain improved results when inputs are parameterised according to the similarity of observations between consecutive time steps.
This parameterisation allows to exclude inputs with rapid and heavy changes, which usually lead to the worst-case bounds but might be rather artificial in certain scenarios.
\end{abstract}

\section{Introduction}
Consider a system consisting of a huge amount of nodes such as a distributed sensor network.
Each node continuously observes its environment and measures information such as temperature, pollution or similar parameters.
Given such a system, we are interested in aggregating information and continuously monitoring properties describing the current status of the system at a central server.
To keep the server's information up to date, the server and the nodes can communicate with each other.
In sensor networks, however, the amount of such communication is particularly crucial, as communication translates to energy consumption, which determines the overall lifetime of the network due to limited battery capacities.
Therefore, algorithms aim at minimizing the communication required for monitoring the respective parameter at the server.

One very basic parameter is the domain of the system defined to be the values currently observed across all nodes.
We consider different notions related to the domain and propose algorithms for monitoring the domain itself, (approximations of) its size and (approximations of) the frequencies of values comprising the domain, respectively.
Each of these parameters can provide useful information, e.g.
the information about the (approximated) frequency of each value allows to approximate very precisely the histogram of the observed values, and
this allows to determine (approximations of) several functions of the input, e.g.\ heavy hitters, quantiles, top-$k$, frequency moments or threshold problems.

\subsection{Model and Problems} 
\label{DR:section:goal}
We consider the continuous distributed monitoring setting, introduced by Cormode, Muthukrishnan, and Yi in \cite{cormodeModel}, in which there are $n$ distributed nodes, each uniquely identified by an identifier (ID) from the set $ \set{1,\dots,n}$, connected to a single server.
Each node observes a stream of values over time and at any discrete time step $t$ node $i$ observes one value $v_i^t \in \set{1,\dots,\Delta}$.
The server is asked to, at any point $t$ in time, compute an output $f(t)$ which depends on the values $v_i^{t'}$ (for $t' \leq t$, and $i = 1, \ldots, n$) observed across all distributed streams up to the current time step $t$.
The exact definition of $f(\cdot)$ depends on the concrete problems under consideration, which are defined in the section below.
For the solution of these problems, we are usually interested in approximation algorithms.
An $\varepsilon$-approximation of $f(t)$ is an output $\tilde f(t)$ of the server such that $(1-\varepsilon)f(t) \leq \tilde f(t) \leq (1+\varepsilon)f(t)$.
We call an algorithm that, for each time step, provides an $\varepsilon$-approximation with probability at least $1-\delta$, an $(\varepsilon, \delta)$-approximation algorithm.
To be able to compute the output, the nodes and the server can communicate with each other by exchanging single cast messages or by broadcast messages sent by the server and received by all nodes.
Both types of communication are instantaneous and have unit cost per message.
That is, sending a single message to one specific node incurs cost of one and so does one broadcast message.
Each message has a size of $O(\log \Delta + \log n + \log\log\frac{1}{\delta})$ bits and will usually, besides a constant number of control bits, consist of a value from $\{1, \ldots, \Delta\}$, a node ID and an identifier to distinguish between messages of different instances of an algorithm applied in parallel (as done when using standard probability amplification techniques).
Having a broadcast channel is an extension to \cite{cormodeModel}, which was originally proposed in \cite{cormodeBroadcast} and afterwards applied in \cite{ipdps1,ipdps2}.
For ease of presentation, we assume that not only the server can send broadcast messages, but also the nodes.
This changes the communication cost only by a factor of at most two, as a broadcast by a node can always be implemented by a single cast message followed by a broadcast of the server.
Between any two time steps we allow a communication protocol to take place, which may use polylogarithmic $\bigO(\log^c n)$ rounds, for some constant $c$.
The optimisation goal is the minimisation of the communication cost, given by the number of exchanged messages, required to monitor the considered problem.

\subsubsection{Monitoring of Domain-Related Functions.}
In this paper, we consider the monitoring of different problems related to the \emph{domain} of the network.
The domain at time $t$ is defined as $D_t \coloneqq \{ v \in \{1, \ldots, \Delta \} \mid \exists i \text{ with } v_i^t = v \}$, the set of values observed by at least one node at time $t$.
We study the following three problems related to the domain:
\begin{itemize}[noitemsep]
    \item[$\bullet$] \textbf{Domain Monitoring.}
    At any point in time, the server needs to know the domain of the system as well as a \emph{representative} node for each value of the domain.
    Formally, monitor $D_t = \{v_1, \ldots, v_{|D_t|}\} \subseteq \{1,\ldots,\Delta\}$, at any point $t$ in time.
    Also, maintain a sequence $R_t = (j_1, \ldots, j_\Delta)$ of nodes such that for all observed values $v \in D_t$ a representative $i$ is determined with $j_v = i$ and $v_i^t = v$. 
    For each value $v \notin D_t$ which is not observed, no representative is given and $j_v = \text{nil}$.
    \item[$\bullet$] \textbf{Frequency Monitoring.} For each $v \in D_t$ monitor the frequency $|N_t^v|$ of nodes in $N_t^v \coloneqq \{i \in \{1,\ldots, n\}  \mid v^t_i = v\}$ that observed $v$ at $t$, i.e.\ the number of nodes currently observing $v$.
    \item[$\bullet$] \textbf{Count Distinct Monitoring.} Monitor $|D_t|$, i.e. the number of distinct values observed at time t.
\end{itemize}
We provide an exact algorithm for the Domain Monitoring Problem and $(\varepsilon, \delta)$-approximations for the Frequency and Count Distinct Monitoring Problem.

\subsection{Our Contribution}
For the Domain Monitoring Problem, an algorithm which uses $\Theta(\sum_{t \in T}|D_t|)$ messages on expectation for $T$ time steps is given in \cref{se:domain}.
This is asymptotically optimal in the worst-case in which $D_t \cap D_{t+1} = \emptyset$ holds for all $t \in T$.
We also provide an algorithm and an analysis based on the minimum possible number $R^*$ of changes of representatives for a given input.
It exploits situations where $D_t \cap D_{t+1} \neq \emptyset$ and uses $\bigO(\log n \cdot R^*)$ messages on expectation.

For an \edapproximation of the Frequency Monitoring Problem for $T$ time steps, we first provide an algorithm using $\Theta(\sum_{t \in T}|D_t|\frac{1}{\varepsilon^2} \log \frac{|D_t|}{\delta})$ messages on expectation in \cref{sec:frequencies}.
We then improve this bound for instances in which observations between consecutive steps have a certain similarity.
That is, for inputs fulfilling the property that for all $v\in\{1,\ldots,\Delta\}$ and some $\sigma \leq 1/2$, the number of nodes observing $v$ does not change by a factor larger than $\sigma$ between consecutive time steps,
we provide an algorithm that uses an expected amount of $\bigO(|D_1| (\max(\delta, \sigma)T+1) \frac{1}{\varepsilon^2} \log \frac{|D_1|}{\delta})$ messages.
\todo{stimmt die schranke da jetzt so? Hab das $|D_1|$ innen log gezogen}
In \cref{sec:countDistinct}, we provide an algorithm using $\Theta(T \cdot \frac{1}{\varepsilon^2} \log \frac{1}{\delta})$ messages on expectation for the Count Distinct Monitoring Problem for $T$ time steps.
For instances which exhibit a certain similarity an algorithm is presented which monitors the problem using 
$\Theta \left( \left( 1+T\cdot \max\{2\sigma, \delta \} \right) \frac{\log (n) \cdot R^*}{|D_t| \cdot \varepsilon^2} \log \frac{1}{\delta} \right)$
%$\Theta \left( \left(T\cdot \max\{\sigma, \delta \} +1\right) \frac{1}{\varepsilon^2} \log \frac{1}{\delta} \right)$ 
messages on expectation.

\subsection{Related Work}
\label{DR:section:relatedWork}
The basis of the model considered in this paper is the \emph{continuous monitoring model} as introduced by Cormode, Muthukrishnan and Yi in \cite{cormodeModel}.
In this model, there is a set of $n$ distributed nodes each observing a stream given by a multiset of items in each time step.
The nodes can communicate with a central server, which in turn has the task to continuously, at any time $t$, compute a function $f$ defined over all data observed across all streams up to time $t$.
The goal is to design protocols aiming at the minimisation of the number of bits communicated between the nodes and the server.
In \cite{cormodeModel}, the monitoring of several functions is studied in their (approximate) threshold variants, in which the server has to output $1$ if $f\geq \tau$ and $0$ if $f \leq (1-\varepsilon)\tau$, for given $\tau$ and $\varepsilon$.
Precisely, algorithms for the frequency moments $F_p= \sum_i m_i^p$ where $m_i$ denotes the frequency of item $i$ for $p=0,1,2$ are given.
$F_1$ represents the simple sum of all items received so far and $F_0$ the number of distinct items received so far.
Since the introduction of the model, monitoring of several functions has been studied such as the monitoring of frequencies and ranks by Huang, Yi and Zhang in \cite{huang}.
The frequency of an item $i$ is defined to be the number of occurrences of $i$ across all streams up to the current time.
The rank of an item $i$ is the number of items smaller than $i$ observed in the streams.
Frequency moments for any $p>2$ are considered by Woodruff and Zhang in \cite{woodruff}.
A variant of the Count Distinct Monitoring Problem is considered by Gibbons and Tirthapura in \cite{gibbons}.
The authors study a model in which each of two nodes receives a stream of items and at the end of the streams a server is asked to compute $F_0$ based on both streams.
A main technical ingredient is the use of so called public coins, which, once initialized at the nodes, provide a way to let different nodes observe identical outcomes of random experiments without further communication.
We will adopt this technique in \cref{sec:countDistinct}.
Note that the previously mentioned problems are all defined over the items \emph{received so far}, which is in contrast to the definition of monitoring problems which we are going to consider and which are all defined only based on the \emph{current time step}.
This fact has the implication that in our problems the monitored functions are no longer monotone, which makes its monitoring more complicated.

Concerning monitoring problems in which the function tracked by the server only depends on the current time step, there is also some previous work to mention.
In \cite{lam}, Lam, Liu and Ting study a setting in which the server needs to know, at any time, the order type of the values currently observed.
That is, the server needs to know which node observes the largest value, second largerst value and so on at time $t$.
In \cite{yi}, Yi and Zhang consider a system only consisting of one node connected to the server.
The node continuously observes a $d$-dimensional vector of integers from $\{1, \ldots, \Delta\}$.
The goal is to keep the server informed about this vector up to some additive error per component.
In \cite{davis}, Davis, Edmonds and Impagliazzo consider the following resource allocation problem:
%We describe it in our terminology): 
$n$ nodes observe  streams of required shares of 
%  that describe an amount of a shared
a given resource.
The server has to assign, to each node, in each time step, a share of the  resource that is as least as large as the required share.
The objective is then given by the minimization of communication necessary for adapting the assignment of the resource over time.

\section{The Domain Monitoring Problem}
\label{se:domain}
We start by presenting an algorithm to solve the Domain Monitoring Problem for a single time step. 
We analyse the communication cost using standard worst-case analysis and show tight bounds.
By applying the algorithm for each time step, we then obtain tight bounds for monitoring the domain for any $T$ time steps.
The basic idea of the protocol as given in \cref{alg:p1} is quite simple:
Applied at a time $t$ with a value $v \in \{1,\ldots, \Delta\}$, the server gets informed whether $v \in D_t$ holds or not.
To do so, each node $i$ with $v_i^t=v$ essentially draws a value from a geometric distribution and then those nodes having drawn the largest such value send broadcast messages.
By this, one can show that on expectation only a constant number of messages is sent.

Furthermore, if applied with $v = nil$, the server can decide whether $v' \in D_t$ for all $v' \in \{1,\ldots,\Delta\}$ at once with $\Theta(|D_t|)$ messages on expectation.
To this end, for \emph{each} $v' \in \{1,\ldots,\Delta\}$ independently, the nodes $i$ with $v_i^t=v'$ drawing the largest value from the geometric distribution send broadcast messages.
In the presentation of \cref{alg:p1}, we assume that $v^t_i = v$ is always true if $v = nil$.
Also, in order to apply it to a subset of nodes, we assume that each node maintains a value $status_i \in \{0,1\}$ and only nodes $i$ take part in the protocol for which $status_i = status$ holds.

\begin{algorithm}[ht]
\begin{enumerate}[noitemsep]
    \item Each node $i$ for which $status_i = status$ and $\left(v \neq nil\Rightarrow v^t_i = v\right)$ hold, draws a value $\hat h_i$ from a geometric distribution with success probability $p \coloneqq 1/2$.
    \item Let $h_i = \min\{\log n, \hat h_i \}$.
    \item Node $i$ broadcasts its value in round $\log n -h_i$ unless a node $i'$ with $v^t_i = v^t_{i'}$ has broadcasted before.
\end{enumerate}
\caption{\protone($v, status$) \hfill\textit{[for fixed time $t$]}}
\label{alg:p1}
\vspace{-3mm}
\end{algorithm}

We have the following lemma, which bounds the expected communication cost of \cref{alg:p1} and has already appeared in a similar way in \cite{ipdps2} (Lemma~III.1).

\begin{lemma}
 \label{lemma:protone_message_complexity}
Applied for a fixed time $t$, \protone($v, 1$) uses $\Theta(1)$ messages on expectation if $v \neq nil$ and $\Theta(|D_t|)$ otherwise.
\end{lemma}
\begin{proof}
	First consider the case where $v \neq nil$.
	Regarding the expected communication of \protone($v, 1$) we introduce some notation.
	Let $ X_{i} $ be a $ \set{0,1} $-random variable indicating whether the node $i \in \ntv $ sends a message to the server, % and $ h_{i} $ the height of sensor $ i $.
	and $X\coloneqq \sum X_i$. %Furthermore let $ X $ be the sum of all $ X_{i} $'s.
	According to the algorithm a sensor $ i $ sends a message if and only if its height $ h_{i} $ matches the round specified for that height and no other sensor $ i' $ has sent its value beforehand.
	We obtain
	\begin{align*}
	\Prob{X_{i} = 1} 
	&= \Prob{\exists r \in \set{1,\dots,\log n}: h_{i} = r \land \forall i' \in \ntv \setminus\set{i} : h_{i'} \leq r  }\\
	&\leq \sum_{r=1}^{\log n} \frac{1}{2^r}\left(1-\frac{1}{2^r}\right)^{n^v-1} .
	\end{align*}
	We know that $ \textnormal{E}[X_{i}] = \Prob{X_{i}=1} $ and thus
	\begin{align*}
	\textnormal{E}[X] 
	\leq n^v \cdot \sum_{r=1}^{\log n} \frac{1}{2^r}\left(1-\frac{1}{2^r}\right)^{n^v-1}.
	\end{align*} 
	Observing that $f(r) = n^v \cdot \frac{1}{2^{r}}\paren{1-\frac{1}{2^{r}}}^{n^v-1}$ has only one extreme point and $f(r)\leq 2$ for all $r\in[0,\log(n)]$, we use the integral test for convergence to obtain
	\begin{align*}
	\textnormal{E}[X]
	&\leq n^v \cdot\sum_{r=1}^{\log n}\frac{1}{2^{r}}\paren{1-\frac{1}{2^{r}}}^{n^v-1}
	\leq n^v \int_{0}^{\log n}\frac{1}{2^{r}}\paren{1-\frac{1}{2^{r}}}^{n^v-1}\textnormal{dr} + 2\\
	&\leq \left[\frac{1}{\ln\paren{2}}\paren{1-\frac{1}{2^{r}}}^{n^v}\right]_{0}^{\log n}+2
	\leq \frac{1}{\ln\paren{2}}+2
	<4.
	\end{align*}
	
	For the case $v= nil$ we can apply the same argumentation independently for each value $v \in D_t$.
	This concludes the proof of the lemma.
	\hfill %\qed
\end{proof}

In order to solve the domain monitoring problem for $T$ time steps, the server proceeds as follows:
In each step $t$ the server calls \protone($nil, 1)$ to identify all values belonging to $D_t$ as well as a valid sequence $R_t$.
By the previous lemma we then have an overall communication cost of $\Theta(|D_t|)$ for each time step $t$.
For monitoring $T$ time steps, the cost is $\Theta(\sum_{t \in T} |D_t|)$.
This is asymptotically optimal in the worst-case since on instances where $D_t \cap D_{t+1} = \emptyset$ for all $t$, any algorithm has cost $\Omega(\sum_{t \in T} |D_t|)$.

\begin{theorem}
	Using \protone($v, 1$), the Domain Monitoring Problem for $T$ time steps can be solved using  $\Theta(\sum_{t \in T}|D_t|)$ messages on expectation.
\end{theorem}

%\newpage
\subsection*{A Parameterised Analysis}
Despite the optimality of the result, the strategy of computing a new solution from scratch in each time step seems unwise and the analysis does not seem to capture the essence of the problem properly.
It often might be the case that there are some similarities between values observed in consecutive time steps and particularly, that $D_t \cap D_{t+1} \neq \emptyset$.
In this case, there might be the chance to keep a representative for several consecutive time steps, which should be exploited. % by an algorithm.
Due to these observations we next define a parameter describing this behavior and provide a parameterised analysis.
To this end, we consider the number of component-wise differences in the sequences of nodes $R_{t-1}$ and $R_t$ and call this difference the \emph{number of changes of representatives} in time step $t$.
Let $R^*$ denote the minimum possible number of changes of representatives (over all considered time steps $T$).
The formal description of our algorithm is given in \cref{alg:domainMonitoring}.
Roughly speaking, the algorithm defines, for each value $v$, phases, where a phase is defined as a maximal time interval during which there exists one node observing value $v$ throughout the entire interval.
Whenever a node being a representative for $v$ changes its observation, it informs the server so that a new representative can be chosen (from those observing $v$ throughout the entire phase, which is indicated by $status_i=1$).
If no new representative is found this way, the server tries to find a new representative among those observing $v$ and for which $status_i=0$ and ends the current phase.
Additionally, if a node observes a value $v$ at time $t$ for which $v \notin D_t$, a new representative is determined among these nodes.
Note that this requires each node to store $D_t$ at any time $t$ and hence a storage of $\bigO(\Delta)$.

\begin{algorithm}[ht]
\small
\textbf{(Node $i$)}
\vspace{-1mm}
\begin{enumerate}[noitemsep]
    \item Define $status_i \coloneqq 1$.
    \item If at some time $t$, $v^t_i \neq v^{t-1}_{i}$, then
    \begin{enumerate}[noitemsep]
        \item If $v^t_i \notin D_{t-1}$, set $status_i = 0$ and apply \protone($v^t_i, 0$).
        \item If $v^t_i \in D_{t-1}$, set $status_i = 0$. Additionally inform server in case $i \in R_{t-1}$.
        \item If server starts a new phase for $v = v^t_i$, set $status_i = 1$.
    \end{enumerate}
\end{enumerate}

\textbf{(Server)}

\vspace{2mm}
\textit{[Initialisation]}

Call \protone($nil, 1)$ to define $D_0$ and for each $v \in D_0$ choose a representative uniformly at random from all nodes which have sent $v$.

\vspace{2mm}
\textit{[Maintaining $D_t$ and $R_t$ at time $t$]}

Start with $D_t = D_{t-1}$ and $R_t = R_{t-1}$ and apply the following rules:\vspace{-2mm}
\begin{enumerate}[noitemsep]
	\item[$\bullet$] \textit{[Current Phase, (try to) find new representative]}
    \item[] If informed by representative of a value $v \in D_{t-1}$,
    \begin{enumerate}[label=\arabic*)]
     \item Call \protone($v, 1)$.
     \item If node(s) respond(s), choose new representative among the responding sensors uniformly at random. \label{step:changeRep}
     \item Else call \protone($v, 0)$. End current phase for $v$ and, if there is no response, delete $v$ from $D_t$ and the respective representative from $R_t$. \label[step]{step:removed}
    \end{enumerate}
    \item[$\bullet$] \textit{[If \protone($v,0$) leads to received message(s), start new phase]}
    \item[]
     Start a new phase for value $v$ if message from an application of \protone($v, 0)$ (by \cref{step:removed} initialised by the server or initialised in Step 2.1.\ by a node) is received.
    Add or replace respective representative in $R_t$ by choosing a node uniformly at random from those responding to \protone($v, 0)$. \label{step:restart}
\end{enumerate}
\caption{\DoMon}
\label{alg:domainMonitoring}
\vspace{-3mm}
\end{algorithm}

\newpage
\begin{theorem}
\textsc{DomainMonitoring} as described in \cref{alg:domainMonitoring} solves the Domain Monitoring Problem using $\bigO(\log n \cdot R^*)$ messages on expectation, where $R^*$ denotes the minimum possible number of changes of representatives.
\end{theorem}

\begin{proof}
We consider each value $v \in \bigcup_t D_t$ separately.
Let $N_{t_1, t_2} \coloneqq \{i \mid v_i^t = v \, \forall t_1 \leq t \leq t_2\}$ denote the set of nodes that observe the value $v$ at each point in time $t$ with $t_1 \leq t\leq t_2$. 
Consider a fixed phase for $v$ and let $t_1$ and $t_2$ be the points in time where the phase starts and ends, respectively.
A phase only ends in \cref{step:removed}, hence there was no response from \protone$(v,1)$, which implies $N^v_{t_1,t_2} = \emptyset$.
Thus, to each phase for $v$ we can associate a cost of at least one to $R^*$ and this holds for each $v \in \bigcup_t D_t$.
Therefore, $R^*$ is at least the overall number of phases of all values.

Next we analyze the expected cost of \cref{alg:domainMonitoring} during the considered phase for $v$.
Let w.l.o.g.\ $N_{t_1} \coloneqq N_{t_1,t_1}=  \{1, 2, \ldots, k\}$.
With respect to the fixed phase, only nodes in $N_{t_1}$ can communicate and the communication is bounded by the number of changes of the representative for $v$ during the phase.
Let $t'_i$ be the first time after $t_1$ at which node $i$ does not observe $v$.
Let the nodes be sorted such that $i<j$ implies $t'_i \geq t'_j$.
Let $a_1, \ldots, a_m$ be the nodes \cref{alg:domainMonitoring} chooses as representatives in the considered phase.
We want to show that $\mathbb{E}[m] = \bigO( \log k)$.
To this end, partition the set of time steps $t'_i$ into groups $G_i$.
Intuitively, $G_i$ represents the time steps in which  the nodes continuously observe value $v$ since time $t_1$ and the size of the initial set of nodes that observed $v$ is halved $i$ times.
Formally, $G_i$ contains all time steps $t_{\ell_{i-1}+1},\ldots,t_{\ell_i}$ (where $\ell_{-1}\coloneqq 0$ for convenience) such that $\ell_i$ is the largest integer fulfilling $|N_{t_1, t'_{\ell_i}}| \in (k/2^{i+1}, k/2^i]$.

Let $S_i$ be the number of changes of representatives in time steps belonging to $G_i$.
We have $\mathbb{E}[m] = \sum_{i=0}^{\log k} \mathbb{E}[S_i]$.
Consider a fixed $S_i$. 
Let $\mathcal{E}_j$ be the event that the $j$-th representative chosen in time steps belonging to $G_i$ is the first one with an index in $\left\{1, \ldots, \lfloor\frac{k}{2^{i+1}}\rfloor\right\}$.
Observe that as soon as this happens, the respective representative will be the last one chosen in a time step belonging to group $G_i$. 

Now, since the algorithm chooses a new representative uniformly at random from the index set $\left\{1, \ldots, \lfloor\frac{k}{2^{i}}\rfloor\right\}$,
the probability that it chooses a representative from $\left\{1, \ldots, \lfloor\frac{k}{2^{i+1}}\rfloor\right\}$ is at least $1/2$ except for the first representative of $v$, where it might be slightly smaller due to rounding errors.
$\mathcal{E}_j$ occurs only if the first $j-1$ representatives were each \emph{not} chosen from this set, i.e. $\Prob{\mathcal{E}_j}\leq \left(\frac{1}{2}\right)^{j-2}$.
Hence,
$\E[S_i] = \sum_{j} \E[S_i | \mathcal{E}_j] \cdot \Pr[\mathcal{E}_j] 
\leq \sum_j j \cdot (\frac{1}{2} )^{j-2} 
= \sum_j \frac{j}{2^{j-2}} = \bigO(1)$.
%\qed
\end{proof}

\section{The Frequency Monitoring Problem}
\label{sec:frequencies}
In this section we design and analyse an algorithm for the Frequency Monitoring Problem, i.e.\ to output (an approximation) of the number of nodes currently observing value $v$.
We start by considering a single time step and present an algorithm which solves the subproblem to output the number of nodes that observe $v$ within a constant multiplicative error bound. 
Afterwards, and based on this subproblem, a simple sampling algorithm is presented which solves the Frequency Monitoring Problem for a single time step up to a given (multiplicative) error bound and with demanded error probability. 

While in the previous section we used the algorithm \protone\ with the goal to obtain a representative for a measured value, in this section we will use the same algorithm to estimate the number of nodes that measure a certain value $v$.
Observe that the expected maximal height of the geometric experiment increases with a growing number of nodes observing $v$.
We exploit this fact and use it to estimate the number of nodes with value $v$, while still expecting constant communication cost only.
For a given a time step $t$ and a value $v \in D_t$, we define an algorithm \CoFac\ as follows:
We apply \protone($v,1$) with $status_i=1$ for all nodes $i$. 
If the server receives the first response in communication round $r \leq \log n$, the algorithm outputs $\tilde n_{\text{const}}^v=2^r$ as the estimation for $\ntvsize$.

We show that we compute a constant factor approximation with constant probability. Then we amplify this probability using multiple executions of the algorithm and taking the median (of the executions) as a final result.
\begin{lemma}
\label{lemma:constant_factor_approximation}
The algorithm \CoFacS estimates the number $\ntvsize$ of nodes observing the value $v$ at time $t$ up to a factor of $8$, i.e.\ $\tilde n^{v}_{const} \in [\ntvsize/8, \ntvsize \cdot 8]$ with constant probability.
\end{lemma}

\begin{proof}
Let $n^v$ be the number of nodes currently observing value $v$, i.e.\ $n^v \coloneqq \ntvsize$.
Recall that the probability for a single node to draw height $h$ is $\Pr[h_i = h]=\frac{1}{2^h}$, if $h < \log n$, and $\Pr[h_i = h] = \frac{2}{2^{h}}$, if $h = \log n$.
Hence, $\Pr[h_i \geq h]=\frac{1}{2^{h-1}}$ for all $h\in\{1,\ldots,\log n\}$.

We estimate the probability of the algorithm to fail, by analysing the cases that $\tilde{n}^v_{const}$ is larger than $\log n^v+3$ or smaller than $\log n^v-3$.
We start with the first case and by applying a union bound we obtain:
\begin{align*}
\Pr[\exists i : h_i > \log n^v + 3] 
 &\leq \Pr[\exists i : h_i \geq \lceil\log n^v\rceil + 3] \\
 &= n^v \cdot \left(\frac{1}{2} \right)^{\lceil\log n^v\rceil + 2}
\leq\frac{1}{4}.
\end{align*}

For the latter case we bound the probability that each node has drawn a height strictly smaller than $\log n^v - 3$ by
\begin{align*}
\Pr[\forall i :h_i < \log n^v - 3] 
 & \leq \prod_{i} \Pr[h_i < \lceil\log n^v\rceil - 3]\\
 &= \left(1- \frac{1}{2^{\lceil\log n^v\rceil - 4}} \right)^{n^v}
\leq \left(1- \frac{8}{n^v} \right)^{n^v} 
\leq \frac{1}{e^8}.
\end{align*}

Thus, the probability that we compute an 8-approximation is bounded by
\begin{align*}
  \Pr \left[\frac{n^v}{8} \leq 2^{h_i} \leq 8 n^v \right]
  &=1- \bigl( \Pr[\exists i: h_i >\log n^v + 3]  +  \Pr[ \forall i:h_i < \log n^v - 3] \bigr)\\
  &\geq 1- \left(\frac{1}{4} + \frac{1}{e^8} \right)
  > 0.7
\end{align*}
%\vspace{-5mm}
\end{proof}

We apply an amplification technique to boost the success probability to arbitrary $1-\delta'$ using $\Theta(\log \frac{1}{\delta'})$ parallel executions of the \CoFacS algorithm and choose the median of the intermediate results as the final output. 

\begin{corollary}
 \label{corollary:constant_factor_approximation_delta}
 Applying $\Theta\left( \log \frac{1}{\delta'} \right)$ independent, parallel instances of \CoFacS, we obtain a constant factor approximation of $\ntvsize$ with success probability at least $1-\delta'$ using $\Theta \left( \log \frac{1}{\delta'} \right)$ messages on expectation.
\end{corollary}
\begin{proof}
	%By a standard amplification argument we boost the probability to an arbitrary success probability of $1-\delta$ using 
	Choose $d=\frac{45}{2}\ln \frac{1}{\delta'}$ to be the number of copies of the algorithm and return the median of the intermediate results. 
	Let $\mathcal{I}_j$ be the indicator variable for the event that the $j$-th experiment does not result in an 8-approximation.
	By \cref{lemma:constant_factor_approximation} the failure probability can be upper bounded by a constant, i.e.\ $\Prob{\mathcal{I}_j}\leq 0.3$.
	Hence, using a Chernoff bound, the probability that at least half of the experiments do meet the required approximation factor of $8$ is
	\begin{align*}
	\Prob{\sum_{j=1}^d \mathcal{I}_j \geq \frac{1}{2}d}
	&\leq\Prob{\sum_{j=1}^d \mathcal{I}_j \geq \left(1+\frac{2}{3} \right)\cdot 0.3\cdot d}\\
	&\leq e^{-\left(\frac{2}{3}\right)^2\cdot\frac{1}{3}\cdot 0.3\cdot d}
	=e^{-\frac{2}{45}\cdot d}
	=e^{-\frac{2}{45}\cdot \frac{45}{2}\ln \frac{1}{\delta'} }
	=\delta'.
	\end{align*}
	Observe that if at least half of the intermediate results are within the demanded error bound, so is the median. 
	Thus, the algorithm produces an $8$-approximation of $\ntvsize$ with success-probability of at least $1-\delta'$, concluding the proof.
	%\qed
\end{proof}

To obtain an $(\varepsilon, \delta)$-approximation, in \cref{alg:epsilon_factor_approximation} we first apply the \CoFacS algorithm to obtain a rough estimate of $\ntvsize$. 
It is used to compute a probability $p$, which is broadcasted to the nodes, so that every node observing value $v$ sends a message with probability $p$.
Since the \CoFacS result $\tilde n_{\text{const}}^v$ in the denominator of $p$ is close to $\ntvsize$, the number of messages sent on expectation is independent of $\ntvsize$.
The estimated number of nodes observing $v$ is then given by the number of responding nodes $\bar n^v$ divided by $p$, which, on expectation, results in $\ntvsize$.

\begin{algorithm}[ht]
	\textbf{(Node $i$)}\vspace{-3mm}
	\begin{enumerate}[noitemsep]
		\item Receive $p$ from the server.
		\item Send a response message with probability $p$.
	\end{enumerate}

	\textbf{(Server)}\vspace{-3mm}
	\begin{enumerate}[noitemsep]
		\item Set $\delta' \coloneqq \frac{\delta}{3}$
		\item Call \CoFac($v$, $\delta'$) to obtain $\tilde n^{v}_{\text{const}}$.
		\item Broadcast $p=\min \left( 1, \frac{24}{\varepsilon^2 \tilde n^{v}_{\text{const}}}\cdot \ln \frac{1}{\delta'} \right)$. \label[step]{setp}
		\item Receive $\nbarv$ messages.
		\item Compute and output estimated number of nodes in $\ntv$ as $\ntildev = \nbarv / p$.
	\end{enumerate}
	\caption{\EpsFac($v \in D_t$, $\varepsilon, \delta$) \hfill\textit{[for fixed time $t$]}}
	\label{alg:epsilon_factor_approximation}
	\vspace{-3mm}
\end{algorithm}

\begin{lemma}
 \label{lemma:single_shot_epsilon}
 The algorithm \EpsFacS as given in \cref{alg:epsilon_factor_approximation} provides an \edapproximation of $\ntvsize$.
\end{lemma}

\begin{proof}
The algorithm obtains a constant factor approximation $\tilde n^{v}_{\text{const}}$ with probability $1-\delta'$.
The expected number of messages is
$\E\left[\nbarv\right]=n^v \cdot p$.

We start by estimating the conditional probability that more than $(1+\varepsilon)\, n^v p$ responses are sent under the condition that $\tilde n^{v}_{\text{const}}\leq8n^v$ and $p<1$.
In this case we have $$p = \frac{24}{\varepsilon^2 \tilde n^{v}_{\text{const}}}\cdot \ln \frac{1}{\delta'} \geq \frac{3}{\varepsilon^2 n^v}\cdot \ln \frac{1}{\delta'},$$ hence using a Chernoff bound it follows
\begin{align*}
p_1\coloneqq \Prob{\nbarv \geq (1+\varepsilon) n^v p \left| \tilde n^{v}_{\text{const}}\leq8n^v\wedge p<1\right.}
 \leq e^{-\frac{\varepsilon^2}{3} n^v \cdot \frac{3}{\varepsilon^2 n^v}\cdot \ln \frac{1}{\delta'}}
 = \delta'.
\end{align*}
Likewise the probability that less than $(1-\varepsilon)\, n^v p$ messages are sent under the condition that $\tilde n^{v}_{\text{const}}\leq8n^v$ and $p<1$ is %the constant approximation succeeded is %(and thus $p \leq \frac{3}{\varepsilon^2 n^v}\cdot \ln \frac{1}{\delta'}$)%(again, if $p<1$) is
\begin{align*}
p_2 &\coloneqq \Prob{\nbarv \leq (1-\varepsilon)n^v p\left| \tilde n^{v}_{\text{const}}\leq8n^v\wedge p<1\right.}\\
%  \leq e^{-\frac{\varepsilon^2}{2} n^v p}
 &\leq e^{-\frac{\varepsilon^2}{2}n^v\cdot \frac{3}{\varepsilon^2 n^{v} }\cdot \ln \frac{1}{\delta'}}
 \leq e^{-\frac{3}{2} \ln \frac{1}{\delta'}}
 < \delta'
 .
\end{align*}
Next consider the case that $\tilde n^{v}_{\text{const}} > 8n^v$ and $p <1$ holds.
Using $$\Prob{\tilde n^{v}_{\text{const}} > 8n^v} \leq \Prob{\tilde n^{v}_{\text{const}} > 8n^v \vee \tilde n^{v}_{\text{const}} < \frac{n^v}{8}} \leq \delta'$$ and $p_i\cdot \Prob{\tilde n^{v}_{\text{const}} \leq 8n^v}\leq p_i$ for $i\in\{1,2\}$,
\begin{align*}
\Prob{(1-\varepsilon)n^v p < \nbarv < (1+\varepsilon) n^v p\left| p < 1\right.} \hspace{3cm} \\
\geq 1 - \left(\Prob{\tilde n^{v}_{\text{const}} > 8n^v} + (p_1+p_2)\right)
				     \geq 1-3\delta' = 1-\delta.
\end{align*}
For the last case $p=1$, we have 
%\begin{align*}
$
 \Prob{(1-\varepsilon)n^v p < \nbarv < (1+\varepsilon)n^v p\left| p \geq 1\right.} = 1,
 $
%\end{align*}
by using $\nbarv = n^v$. 
Now, $\Prob{(1-\varepsilon)n^v p < \nbarv < (1+\varepsilon)n^v p} \geq 1-\delta$ directly follows.
%\qed
\end{proof}

\begin{lemma}
\label{lemma:msg_epsilonApprox}
Algorithm \EpsFacS as given in \cref{alg:epsilon_factor_approximation} uses $\Theta(\frac{1}{\varepsilon^2}\log \frac{1}{\delta})$ messages on expectation.
\end{lemma}
\begin{proof}
	Recall that each of the $n^v$ nodes sends a message with probability $p$, leading to $n^v \cdot p$ messages on expectation.
	First assume that the constant factor approximation was successful, i.e.\ $\frac{n_1}{8} \leq \tilde n^v_\text{const} \leq 8 n_1$.
	If $p<1$, we have 
	$$n^v \cdot p = n^v \frac{24}{\varepsilon^2 \tilde n^{v}_{\text{const}}}\cdot \ln \frac{1}{\delta'} \leq \frac{24\cdot 8}{\varepsilon^2}\cdot \ln \frac{1}{\delta'}=\Theta\left( \frac{1}{\varepsilon^2}\log \frac{1}{\delta} \right).$$
	If $p=1$, by definition %, receive exactly $\nv$ messages. By definition this is only the case if 
	$\frac{24}{\varepsilon^2 \tilde n^{v}_{\text{const}}}\cdot \ln \frac{1}{\delta'} \geq 1$, hence $\tilde n^{v}_{\text{const}} = \bigO\left( \frac{1}{\varepsilon^2}\cdot \log \frac{1}{\delta'} \right) $.
	Thus, $n^v p \leq 8 \tilde n^{v}_{\text{const}} p = \bigO\left( \frac{1}{\varepsilon^2}\cdot \log \frac{1}{\delta'} \right)$.
	
	For the case that the constant factor approximation was not successful,
	note that $\Prob{\tilde n^{v}_{\text{const}} < \frac{1}{8 \cdot 2^i} n^{v}} \leq \frac{1}{e^{2^{i+3}}}$ holds analogously to the calculation in \cref{lemma:constant_factor_approximation}.
	Also, for $\tilde n^{v}_{\text{const}} \geq \frac{1}{8 \cdot 2^i} n^{v}$ and $p<1$, we have 
	$$n^v p \leq 8 \cdot 2^i\cdot \tilde n^{v}_{\text{const}}\cdot \frac{24}{\varepsilon^2 \tilde n^{v}_{\text{const}}}\cdot \ln \frac{1}{\delta}= 2^i\cdot \Theta\left( \frac{1}{\varepsilon^2}\log \frac{1}{\delta} \right).$$
	Similarly, for $p=1$, we have $n^v p \leq 8 \cdot 2^i\cdot \tilde n^{v}_{\text{const}} = 2^i\cdot \Theta\left( \frac{1}{\varepsilon^2}\log \frac{1}{\delta} \right)$ as in this case, $\tilde n^{v}_{\text{const}} = \bigO\left( \frac{1}{\varepsilon^2}\cdot \log \frac{1}{\delta'} \right)$.
	Hence, we can conclude
	\begin{align*}
	\mathbb{E}\left[\nbarv\right] &\leq \Theta\left( \frac{1}{\varepsilon^2}\log \frac{1}{\delta} \right) \cdot \Prob{\tilde n^v_{\text{const}}\geq \frac{1}{8}n^v} \\&\qquad + \sum_{i=0}^\infty \Prob{\frac{1}{8 \cdot 2^{i+1}} n^{v} \leq \tilde n^v_{\text{const}} < \frac{1}{8 \cdot 2^i} n^{v}}\cdot 2^{i+1}\cdot \Theta\left( \frac{1}{\varepsilon^2}\log \frac{1}{\delta} \right)
	\end{align*}
	\begin{align*}
	&\leq \Theta\left( \frac{1}{\varepsilon^2}\log \frac{1}{\delta} \right) \left(1 + \sum_{i=0}^\infty \frac{2^{i+1}}{e^{2^{i+3}}}\right)
	\leq \Theta\left( \frac{1}{\varepsilon^2}\log \frac{1}{\delta} \right) \left(1 + \sum_{i=0}^\infty 2^{i+1-2^{i+3}}\right)\\
	&\leq \Theta\left( \frac{1}{\varepsilon^2}\log \frac{1}{\delta} \right) \left(1 + \sum_{i=0}^\infty 2^{-i}\right) = \Theta\left( \frac{1}{\varepsilon^2}\log \frac{1}{\delta} \right).
	\end{align*}
	%\qed
\end{proof}

\begin{theorem}
\label{theorem:epsilon_factor_approximation}
There exists an algorithm that provides an ($\varepsilon$,$\delta$)-approximation for the Frequency Monitoring Problem for $T$ time steps with an expected number of $\Theta\left(\sum_{t \in T}|D_t|\frac{1}{\varepsilon^2} \log \frac{|D_t|}{\delta}\right)$ messages.
\end{theorem}
\begin{proof}
 In every time step $t$ we first identify $D_t$ by applying \protone\ using $\Theta \left( |D_t| \right)$ messages on expectation.
 On every value $v \in D_t$ we then perform algorithm \EpsFac($v$,$\varepsilon$,$\frac{\delta}{|D_t|}$), resulting in an amount of $\Theta \left( |D_t| \frac{1}{\varepsilon^2} \log \frac{|D_t|}{\delta} \right)$ messages on expectation for a single time step, while achieving a probability (using a union bound) of $1-\frac{|D_0|\delta}{|D_0|}=1-\delta$ that in one time step the estimations for every $v$ are $\varepsilon$-approximations.
 Applied for each of the $T$ time steps, we obtain a bound as claimed.
 \hfill %\qed
\end{proof}

\subsection*{A Parameterised Analysis}
Applying \EpsFacS in every time step is a good solution in worst case scenarios.
But if we assume that the change in the set of nodes observing a value is small in comparison to the size of the set, we can do better.

We extend the \EpsFacS such that in settings where from one time step to another only a small fraction $\sigma$ of nodes change the value they measure, the amount of communication can be reduced, while the quality guarantees remain intact.
We define $\sigma$ such that
\begin{equation*}
% \forall t: \sigma \geq \frac{|(N^v_t \cup N^v_{t+1}) \setminus (N^v_{t+1} \cap N^v_{t})|}{\ntvsize}.
 \forall t: \sigma \geq \frac{| N_{t-1}^v \setminus \ntv| +  | \ntv \setminus N_{t-1}^v|}{\ntvsize}.
\end{equation*}
Note that this also implies that $D_t=D_{t-1}$ holds for all time steps $t$, i.e. the set of measured values stays the same over time.

The extension is designed so that compared to \EpsFac, also in settings with many changes the solution quality and message complexity asymptotically does not increase.
The idea is the following: For a fixed value $v$, in a first time step \EpsFacS is executed (defining a probability $p$ in \cref{setp} of \cref{alg:epsilon_factor_approximation}).
In every following time step, up to $1/\delta$ consecutive time steps, nodes that start or stop measuring a value $v$ send a message to the server with the same probability $p$, while nodes that do not observe a change in their value remain silent.
In every time step $t$, the server uses the accumulated messages from the first time step and all messages from nodes that started measuring $v$ in time steps $2 \dots t$, while subtracting all messages from nodes that stopped measuring $v$ in the time steps $2\dots t$.
This accumulated message count is then used similarly as in \EpsFacS to estimate the total number of nodes observing $v$ in the current time step.
The algorithm starts again if a) $1/\delta$ time steps are over, so that the probability of a good estimation remains good enough, or b) the sum of estimated nodes to start/stop measuring value $v$ is too large. 
The latter is done to ensure that the message probability $p$ remains fitting to the number of nodes, ensuring a small amount of communication, while guaranteeing an $(\varepsilon, \delta)$-approximation.

Let $n_t^+, n_t^-$ be the number of nodes that start measuring $v$ in time step $t$ or that stop measuring it, respectively, i.e.\
$n_t^+ = | \ntv \setminus N_{t-1}^v|, n_t^- = | N_{t-1}^v \setminus \ntv|$, and $\bar n_t^+$ and $\bar n_t^-$ the number of them that sent a message to the server in time step $t$.
In the following we call nodes contributing to $n^+_t$ and $n^-_t$ \emph{entering} and \emph{leaving}, respectively. 

\begin{algorithm}[ht]

\textbf{(Node $i$)}\vspace{-3mm}
\begin{enumerate}[noitemsep]
 \item If $t=1$, take part in \EpsFac\ called in \cref{alg:ContinuousEpsilonApproximation:start} by the server. \label[step]{alg:ContEpsApprox:1}
 \item If $t>1$, broadcast a message with probability $p$ if $v_i^{t-1}=v \wedge v_i^t \neq v$ \\or $v_i^{t-1} \neq v \wedge v_i^t = v$.
\end{enumerate}

\textbf{(Server)}\vspace{-3mm}
\begin{enumerate}[noitemsep]
    \item Set $\delta' \coloneqq \delta^2$.
    \item Set $t \coloneqq 1$ and run \EpsFac($v$, $\varepsilon/3$, $\delta$) to obtain $\bar n_1, p$. \label[step]{alg:ContinuousEpsilonApproximation:start}
    \item Output $\tilde n_1=\frac{\bar n_1}{p}$.
    \item Repeat at the beginning of every new time step $t>1$:
    \begin{enumerate}[noitemsep, leftmargin=7mm]
	\item Receive messages from nodes changing the observed value to obtain $\bar n_t^+$ and $\bar n_t^-$.
	\item Break if $t \geq 1/\delta$ or $\left(\sum_{i=1}^t \bar n_i^+ + \sum_{i=1}^t \bar n_i^-\right)/p \geq \bar n_1/2$.
	\item Output $\tilde n_t=\left(\bar n_1+\sum_{i=1}^t \bar n_i^+ - \sum_{i=1}^t \bar n_i^-\right)/p$.
    \end{enumerate}
    \item Go to \cref{alg:ContinuousEpsilonApproximation:start}.
\end{enumerate}
\caption{\ContEpsFac($v$, $\varepsilon$, $\delta$)}
\label{alg:ContinuousEpsilonApproximation}
\vspace{-3mm}
\end{algorithm}

\begin{lemma}
 \label{lemma:frequencies_multiple_step_correctness}
 For any $v \in D_1$, the algorithm \ContEpsFacS provides an \edapproximation of $\ntvsize$.
%  The algorithm \ContEpsFacS provides an \edapproximation of $\ntvsize$ for any $v \in D_0$ at any time step $t$.
\end{lemma}
\begin{proof}
	By the same arguments as in \cref{lemma:single_shot_epsilon}, we obtain an ($\varepsilon$,$\delta'$)-approxima\-tion of $n_1$.
	In any further time step we compute our estimate over the sum of all received messages ($\bar n_1$, arrivals and departures). 
	If too many nodes change their measured value, we redo a complete estimation of the nodes in $\ntv$.
	
	Recall that $\tilde n_t$ is the random variable giving the estimated number of nodes by the algorithm, and $\tilde n_t^+=\frac{\bar n^+}{p}, \tilde n_t^-=\frac{\bar n^-}{p}$ are the random variables giving the estimated arrivals and departures in that time step.
	We look at any time step $t>1$ where the restart criteria are not met:
	Since $\tilde n_t=\tilde n_1 + \sum_{i=2}^t \left( \tilde n_i^+ - \tilde n_i^- \right)$ and the linearity of expectation,
	for any time $t \geq 1$ we can use a Chernoff bound as in \cref{lemma:single_shot_epsilon} to show that the estimation is an $(\varepsilon,\delta')$-approximation.
	
	Using a union bound on the fail probability of up to $1/\delta$ time steps, we get a $1 - \frac{1}{\delta} \cdot \delta'=1-\delta$ probability of having a correct estimation in any time step.
	\hfill %\qed
\end{proof}

\begin{lemma}
  \label{lemma:frequencies_multiple_step_complexity}
  For a fixed value $v$ and $T'=\min\{\frac{1}{2\sigma},\frac{1}{\delta}\}$, $\sigma \leq \frac{1}{2}$, time steps, \ContEpsFacS uses $\Theta \left(\frac{1}{\varepsilon^2} \log \frac{1}{\delta} \right)$ messages on expectation.
\end{lemma}
 \begin{proof}
 	%  We showed that \cref{alg:ContinuousEpsilonApproximation} gives an $(\varepsilon,\delta)$-approximation for up to $T'=\min\{\frac{1}{2\sigma},\frac{1}{\delta}\}$ time steps in a row.
 	The message complexity depends on the initial size $|N_1^v|$ and on the number of nodes leaving and entering $N^v$ in those time steps, which is bounded by $\sigma$.
 	If \EpsFac\ obtained a correct probability $p$  in \cref{alg:ContEpsApprox:1}, i.e. $p = \Theta(\frac{1}{n_1})$, the expected number of messages (in case $p<1$) is
 	
 	\begin{align*}
 	\mathbb{E} \left[\sum_{t=1}^{T'} \bar n_t \, \middle|\, p = \Theta \left(\frac{1}{n_1} \right)\right]
 	&=\mathbb{E} \left[\bar n_1 + \sum_{i=2}^{T'} \bar n_i^+ + \bar n_i^- \, \middle| \, p = \Theta \left(\frac{1}{n_1} \right)\right]
 	\\
 	&=\left(n_1 + \sum_{i=2}^{T'} n_i^+ + n_i^- \right) p
 	\leq\left(n_1 + T' \sigma n_1\right)p
 	\\
 	&= n_1 \left(1 + T' \sigma \right) \cdot 24 \cdot \frac{1}{\varepsilon^2 \tilde n^v_{\text{const}}} \ln \frac{1}{\delta'}
 	\\
 	&= \Theta \left( \left(1+ \min\left\{\frac{1}{2\sigma},\frac{1}{\delta}\right\} \sigma \right) \cdot 1/\varepsilon^2 \log \frac{1}{\delta} \right).
 	\end{align*}
 	Considering the case where \EpsFacS estimated wrong, the message complexity could increase greatly if the probability $p$ is too large for the actual number of nodes (i.e. an underestimation leads to high message complexity). 
 	But the probability to misestimate by some constant factor (which would increase the message complexity by that factor) decreases exponentially in this factor (as shown in \cref{lemma:msg_epsilonApprox} for \EpsFac), leaving the expected number of messages to be
 	$\Theta \left( \left(1+ \min\left\{\frac{1}{2\sigma},\frac{1}{\delta}\right\} \sigma \right) \cdot \frac{1}{\varepsilon^2}\cdot \log \frac{1}{\delta} \right)
 	=\Theta \left( \frac{1}{\varepsilon^2} \log \frac{1}{\delta} \right)$.
 	\hfill %\qed
 \end{proof}

\begin{theorem}
\label{th:fmptheorem}
  There exists an ($\varepsilon$,$\delta$)-approximation algorithm for the Frequency Monitoring Problem for $T$ consecutive time steps which uses an amount of 
  $\Theta \left( |D_1| \left( 1+T\cdot \max\{2\sigma, \delta \} \right) \frac{1}{\varepsilon^2} \log \frac{|D_1|}{\delta} \right)$ messages on expectation, if $\sigma \leq 1/2$.
\end{theorem}

\begin{proof}
The algorithm works by first applying \protone($nil$,$1$) to obtain $D_1$ and then applying \ContEpsFac($v$, $\varepsilon$, $\delta/|D_1|$) for every $v \in D_1$.
 By \cref{lemma:frequencies_multiple_step_correctness} we know that in every time step and for all $v\in D_1$, the frequency of $v$ is approximated up to a factor of $\varepsilon$ with probability $1-\delta/|D_1|$.
 We divide the $T$ time steps into intervals of size $T'=\min\{\frac{1}{2\sigma},\frac{1}{\delta}\}$ and perform \ContEpsFacS on each of them for every value $v \in D_1$.
 There are $\lceil \frac{T}{T'} \rceil \leq 1+T\cdot \max\{2\sigma,\delta\}$ such intervals.
 For each of those, by \cref{lemma:frequencies_multiple_step_complexity} we need $\Theta \left( \left(1+ \min\{\frac{1}{2\sigma},\frac{1}{\delta}\} \sigma \right) \cdot 1/\varepsilon^2 \log \frac{|D_1|}{\delta} \right)$ messages on expectation for each $v \in D_1$.
 This yields a complexity of $\Theta \left( |D_1| \left( 1+T\cdot \max\{2\sigma, \delta \} \right) \frac{1}{\varepsilon^2} \log \frac{|D_1|}{\delta} \right)$ due to $\min\{\frac{1}{2\sigma},\frac{1}{\delta}\} \sigma\leq \frac{1}{2\sigma} \cdot \sigma = \Theta(1)$.
 Using a union bound over the fail probability for every $v\in D_1$, a success probability of at least $1-\frac{|D_1|\delta}{|D_1|}=1-\delta$ follows.
 \hfill %\qed
\end{proof}

By \cref{theorem:epsilon_factor_approximation}, trivially repeating the single step algorithm \EpsFacS needs $\Theta \left( T |D_1| \frac{1}{\varepsilon^2}\log\frac{|D_1|}{\delta}\right)$ messages on expectation for $T$ (because the number of nodes in $\ntv$ for any $v \in D_1$ is at least $N_1^v/2$ in every time step of that interval).
Hence, the number of messages sent when using \ContEpsFac\ is reduced in the order of $\max\{2\sigma, \delta \}$.

\section{The Count Distinct Monitoring Problem}
\label{sec:countDistinct}

In this section we present an \edapproximation algorithm for the Count Distinct Monitoring Problem.
The basic approach is similar to the one presented in the previous section for monitoring the frequency of each value.
That is, we first estimate $|D_t|$ up to a (small) constant factor and then use the result to define a protocol for obtaining an $(\varepsilon,\delta)$-approximation.
If we could assume that, at any fixed time $t$, each value was observed by at most one node, it would be possible to solve this problem with expected communication cost of $O(\frac{1}{\varepsilon^2} \log \frac{1}{\delta})$ (per time step $t$ and per value $v \in D_t$) using the same approach as in the previous section.
Since this assumption is generally not true, we aim at simulating such behaviour that for each value in the domain only one random experiment is applied. 
We apply the concept of \emph{public coins}, which allows nodes measuring the same value to observe identical outcomes of their random experiments.
To this end, nodes have access to a shared random string $R$ of fully independent and unbiased bits.
This can be achieved by letting all nodes use the same pseudorandom number generator with a common starting seed, adding a constant number of messages to the bounds proven below.
We assume that the server sends a new seed in each phase by only loosing at most a constant factor in the amount of communication used.
However, we can drop this assumption by checking whether there are nodes that changed their value such that only in rounds in which there are changes new public randomness is needed. 
The formal description of the algorithm for a constant factor and an $\varepsilon$-approximation are given in \cref{alg:publicCoin1} and \cref{alg:publicCoin2}, respectively.

We consider the access of the public coin to behave as follows: 
Initialised with a seed, a node accesses the sequence of random bits $R$ bitwise, i.e. after reading the $j$'th bit, the node next accesses bit $j+1$. 
Observe the crucial fact that as long as each node accesses the exact same number of bits, each node observes the exact same random bits simultaneously. 
\cref{alg:publicCoin1} essentially works as follows:
In a first step, each node draws a number from a geometrical distribution using the public coin.
By this, all nodes observing the same value $v$ obtain the same height $h_v$.
In the second step we apply the strategy as in the previous section to reduce communication if lots of nodes observe the same value: Each node $i$ draws a number $g_i$ from a geometrical distribution without using the public coin.
Afterwards, all nodes with the largest height $g_i$ among those with the largest height $h_v$ broadcast their height $h_v$.

\begin{algorithm}[ht]
	\textbf{(Node $i$, observes value $v = v_i$)}
	\vspace{-3mm}
	\begin{enumerate}[noitemsep]
		\item Draw a random number $h_v$ as follows:\\
		Consider the next $\Delta \cdot \log n$ random bits $b_1, \ldots, b_{\Delta \cdot \log n}$ from $R$. 
		Let $h$ be the maximal number of bits $b_{v \cdot \log n + 1}, \ldots, b_{v \cdot \log n + 1+  h}$ that equal $0$.
		Define $h_v \coloneqq \min\{h,\log n\}$.
		\item Let $g'_i$ be a random value drawn from a geometric distribution with success-probability $p = 1/2$ and define $g_i = \min(g'_i, \log n)$ (without accessing public coins).
    	\item Broadcast drawn height $h_v$ in round $r = \log^2 n - (h_v - 1) \cdot \log n - g_i$ unless a node $i'$ has broadcasted before.
	\end{enumerate}

	\textbf{(Server)}
	\vspace{-3mm}
	\begin{enumerate}[noitemsep]
		\item Receive a broadcast message containing height $h$ in round $r$.
		\item Output $\hat{d}_t = 2 ^h$.
	\end{enumerate}
	\caption{\CoFac\hfill\textit{[for fixed time $t$]}}
	\label{alg:publicCoin1}
	\vspace{-3mm}
\end{algorithm}

Note that only (at most $n$) nodes that observe value $v$ with $h_v$ = $max_{v'} \, h_{v'}$ may send a message in \cref{alg:publicCoin1}.
Now, all nodes observing the same value observe the same outcome of their random experiments determining $h_v$.
Hence, by a similar reasoning as in \cref{lemma:constant_factor_approximation}, one execution of the algorithm uses $\bigO(1)$ messages on expectation.

Using the algorithm given in \cref{alg:publicCoin1} and applying the same idea as in the previous section, we obtain an $(\varepsilon, \delta)$-approximation as given in \cref{alg:publicCoin2}:
Each node tosses a coin with a success probability depending on the constant factor approximation (for which we have a result analogous to \cref{corollary:constant_factor_approximation_delta}).
Again, all nodes use the public coin so that all nodes observing the same value obtain the same outcome of this coin flip.
Afterwards, those nodes which have observed a success apply the same strategy as in the previous section, that is, they draw a random value from a geometric distribution, and the nodes having the largest height send a broadcast.

\begin{algorithm}[ht]
	\textbf{(Node $i$)}
	\vspace{-3mm}
	\begin{enumerate}[noitemsep]
		\item Flip a coin with success probability $p = 2^{-q} = \frac{ c \log 1/\delta} {\varepsilon^2 \hat{d}_t}$, $q \in \mathbb{N}$ as follows:\\
		Consider the next $\Delta \cdot q$ random bits $b_1, \ldots b_{\Delta \cdot q}$.
		The experiment is successful if and only if all random bits $b_{v \cdot q + 1}, \ldots, b_{v \cdot q + q}$ equal $0$.
		The node deactivates (and does not take part in Steps 2.\ and 3.) if the experiment was not successful.
		\item Draw a random value $h'_i$ from a geometric distribution and define $h_i = \min(h'_i, \log n)$ (without accessing public coins). 
		\item Node $i$ broadcasts its value in round $\log n - h_i$ unless a node $i'$ with $v^t_i = v^t_{i'}$ has broadcasted before.
	\end{enumerate}
	\textbf{(Server)}
	\vspace{-3mm}
	\begin{enumerate}[noitemsep]
		\item Let $S_t$ be the set of received values.
		\item Output $\tilde{d}_t \coloneqq |S_t| / p$
	\end{enumerate}
	\caption{\EpsFac \hfill\textit{[for fixed time $t$]}}
	\label{alg:publicCoin2}
	\vspace{-3mm}
\end{algorithm}
\newpage
Using arguments analogous to \cref{lemma:single_shot_epsilon,lemma:msg_epsilonApprox} and applying \EpsFacS for $T$ time steps, we obtain the following theorem.

\begin{theorem}
	There exists an $(\varepsilon, \delta)$-approximation algorithm for the Count Distinct Monitoring Problem for $T$ time steps using $\bigO(T \cdot \frac{1}{\varepsilon^2} \log \frac{1}{\delta})$ messages on expectation.
\end{theorem}

\subsection*{A Parameterised Analysis}

In this section we consider the problem for multiple time steps and parameterise the analysis with respect to instances in which the domain does not change arbitrarily between consecutive time steps.
Recall that for monitoring the frequency from a time step $t-1$ to the current time step $t$, all nodes that left and all nodes that entered toss a coin to estimate the number of changes. 
However, to identify that a node observes a value which was not observed in the previous time step, the domain has to be determined exactly. 

We apply the following idea instead: For each value $v \in \{1, \ldots, \Delta\}$ we flip a (public) coin. 
We denote the set of values with a successful coin flip as the \emph{sample}.  
Afterwards, the algorithm only proceeds on the values of the sample, i.e.\ in cases in which a node observes a value with a successful coin flip and no node observed this value in previous time steps, this value contributes to the estimate $\tilde d_t^+$ at time $t$.
Regarding the (sample) of nodes that leave the set of observed values, the \DoMonS algorithm is applied to identify which (sampled) values are not observed any longer (and thus contribute to $\tilde d_t^-$). 
%By this and as long as the domain changes not arbitrarily between consecutive time steps, we can achieve similar bounds as in the previous section.
\begin{algorithm}
	\begin{enumerate}[noitemsep]
		\item Compute $\delta'= 2\,\delta^2$
		\item \label[step]{alg:multiple_step_epsilon:start} 
			Broadcast a new seed value for the public coin. 
		\item Compute an $(\varepsilon, \delta')$-approximation $\tilde d_1$ of $|D_1|$ using \cref{alg:publicCoin2}. 
		Furthermore, obtain the success-probability $p$.
		\item Repeat for each time step $t>1$:
		\begin{enumerate}[noitemsep]
			\item Each node $i$ applies \cref{alg:domainMonitoring} if the observed value $v_i$ is in the sample set.
			Let $\hat d_t^-$ be the number of values (in sample set) which left the domain and $\hat d_t^+$ the number of nodes that join the sample. 

			\item Server computes $\tilde d_t=\tilde d_1 + \sum_{i=2}^t \hat d_i^+/p - \sum_{i=2}^t \hat d_i^-/p$.
			\item Break if $t=1/\delta$ or $\left(\sum_{i=2}^t \tilde d_i^+ + \sum_{i=2}^t \tilde d_i^-\right)/p$ exceeds $\tilde d_1/2$.
		\end{enumerate}
		\item Set $t=1$ and go to \cref{alg:multiple_step_epsilon:start}.
	\end{enumerate}
	\caption{\ContEpsFac($\varepsilon, \delta$)}
	\label{alg:frequencies_multiple_step_epsilon}
	\vspace{-3mm}
\end{algorithm}
Analogous to \cref{lemma:frequencies_multiple_step_correctness}, we have the following lemma.

\begin{lemma}
	\label{lemma:multiple_step_correctness}
	\ContEpsFacS achieves an ($\varepsilon$,$\delta$)-approximation of $|D_t|$ in any time step $t$.
\end{lemma}

For the number of messages, we argue based on the previous section.
However, in addition the \DoMonS algorithm is applied. 
Observe that the size of the domain changes by at most $n/2$, and consider the case that this number of nodes observed the same value $v$. 
The expected cost (where the expectation is taken w.r.t.\ whether $v$ is within the sample) is $\bigO ( \log n \cdot R^* \cdot p) = \bigO \bigl(\frac{\log n \cdot R^*}{|D_t| \varepsilon^2} \log \frac{1}{\delta} \bigr)$.
Similar to \cref{th:fmptheorem}, we then obtain the following theorem.
\begin{theorem}
	\ContEpsFacS provides and $(\varepsilon,\delta)$-approxima\-tion for the Count Distinct Monitoring Problem for $T$ time steps using an amount of $\Theta \left( \left( 1+T\cdot \max\{2\sigma, \delta \} \right) \frac{\log (n) \cdot R^*}{|D_t| \cdot \varepsilon^2} \log \frac{1}{\delta} \right)$ messages on expectation, if $\sigma \leq 1/2$.
\end{theorem}

\bibliographystyle{plain}
\bibliography{bibliography}

\begin{thebibliography}{10}

\bibitem{cormodeModel}
Graham Cormode, S.~Muthukrishnan, and Ke~Yi.
\newblock Algorithms for distributed functional monitoring.
\newblock In {\em Proceedings of the 19th Annual {ACM-SIAM} Symposium on
  Discrete Algorithms (SODA '08)}, pages 1076--1085. {SIAM}, 2008.

\bibitem{cormodeBroadcast}
Graham Cormode, S.~Muthukrishnan, and Ke~Yi.
\newblock {Algorithms for Distributed Functional Monitoring}.
\newblock {\em {ACM} Transactions on Algorithms}, 7(2):21:1--21:20, 2011.

\bibitem{davis}
Sashka Davis, Jeff Edmonds, and Russell Impagliazzo.
\newblock {Online Algorithms to Minimize Resource Reallocations and Network
  Communication}.
\newblock In {\em Proceedings of the 9th International Conference on
  Approximation Algorithms for Combinatorial Optimization Problems, and 10th
  International Conference on Randomization and Computation (APPROX'06/RANDOM
  '06)}, volume 4110 of {\em Lecture Notes in Computer Science}, pages
  104--115. Springer, 2006.

\bibitem{gibbons}
Phillip~B. Gibbons and Srikanta Tirthapura.
\newblock Estimating simple functions on the union of data streams.
\newblock In {\em {Proceedings of the 13th annual ACM Symposium on Parallel
  Algorithms and Architectures (SPAA '01)}}, pages 281--291. ACM, 2001.

\bibitem{huang}
Zengfeng Huang, Ke~Yi, and Qin Zhang.
\newblock Randomized algorithms for tracking distributed count, frequencies,
  and ranks.
\newblock In {\em Proceedings of the 31st {ACM} {SIGMOD-SIGACT-SIGART}
  Symposium on Principles of Database Systems (PODS '12)}, pages 295--306.
  {ACM}, 2012.

\bibitem{lam}
Tak~Wah Lam, Chi{-}Man Liu, and Hing{-}Fung Ting.
\newblock {Online Tracking of the Dominance Relationship of Distributed
  Multi-dimensional Data}.
\newblock In {\em Proceedings of the 8th International Workshop on
  Approximation and Online Algorithms (WAOA '10)}, volume 6534 of {\em Lecture
  Notes in Computer Science}, pages 178--189. Springer, 2010.

\bibitem{ipdps1}
Alexander M{\"{a}}cker, Manuel Malatyali, and Friedhelm {Meyer auf der Heide}.
\newblock {Online Top-k-Position Monitoring of Distributed Data Streams}.
\newblock In {\em Proceedings of the 2015 {IEEE} International Parallel and
  Distributed Processing Symposium (IPDPS '15)}, pages 357--364. IEEE, 2015.

\bibitem{ipdps2}
Alexander M{\"{a}}cker, Manuel Malatyali, and Friedhelm {Meyer auf der Heide}.
\newblock {On Competitive Algorithms for Approximations of Top-k-Position
  Monitoring of Distributed Streams}.
\newblock In {\em Proceedings of the 2016 {IEEE} International Parallel and
  Distributed Processing Symposium (IPDPS '16)}, pages 700--709. IEEE, 2016.

\bibitem{woodruff}
David~P. Woodruff and Qin Zhang.
\newblock Tight bounds for distributed functional monitoring.
\newblock In {\em Proceedings of the 44th Symposium on Theory of Computing
  (STOC '12)}, pages 941--960. {ACM}, 2012.

\bibitem{yi}
Ke~Yi and Qin Zhang.
\newblock Multidimensional online tracking.
\newblock {\em {ACM} Transactions on Algorithms}, 8(2):12, 2012.

\end{thebibliography}

\end{document}